\documentclass[a4paper,11pt]{article}

  \renewcommand\appendix{\par
  \setcounter{section}{0}
  \setcounter{subsection}{0}
  \setcounter{figure}{0}
  \setcounter{table}{0}
  \renewcommand\thesection{ Appendix \Alph{section}}
  \renewcommand\thefigure{\Alph{section}\arabic{figure}}
  \renewcommand\thetable{\Alph{section}\arabic{table}}
}

 \usepackage{amsmath, amsthm, amssymb}
\usepackage{amssymb}
\usepackage{graphicx}
\usepackage{algorithmic}
\usepackage{algorithm}

\usepackage{tikz}
\usetikzlibrary{calc}
\usetikzlibrary{patterns}
\tikzstyle{mybox} = [draw=black, fill=white,  thick,
    rectangle, inner sep=10pt, inner ysep=20pt]
\tikzstyle{mybox} = [draw=black, fill=white,  thick,
    rectangle, inner sep=2pt, inner ysep=2pt]

\usepackage[margin=1.2in]{geometry}

\usepackage{setspace}
\newtheorem{lemma}{Lemma}[section]
\newtheorem{theorem}{Theorem}[section]

\begin{document}

%\mainmatter  % start of an individual contribution

\title{An improved Constant-Factor Approximation Algorithm for Planar Visibility Counting Problem}

%\titlerunning{Visibility Testing and Counting}

\author{Sharareh Alipour  \and Mohammad Ghodsi   \and Amir Jafari}
%\authorrunning{Alipour and ...}


%\institute{Department of Computer engineering\\ Sharif University of Technology \\
%Tehran, Iran\\
%}


%\toctitle{Visibility Testing and Counting}
%\tocauthor{Visibility Testing and Counting}
\maketitle

\begin{abstract}Given a set $S$ of $n$ disjoint line segments in $\mathbb{R}^{2}$, the visibility counting problem (VCP) is to preprocess $S$ such that the number of segments in $S$ visible from any query point $p$ can be computed quickly. This problem can  trivially  be solved in logarithmic query time  using $O(n^{4})$ preprocessing time and space.
Gudmundsson and Morin  proposed a 2-approximation algorithm for this problem with a tradeoff between the space and the query time. They answer any query  in $O_{\epsilon}(n^{1-\alpha})$ with $O_{\epsilon}(n^{2+2\alpha})$ of preprocessing time and space, where $\alpha$ is a constant $0\leq \alpha\leq 1$, $\epsilon > 0$ is another constant that can be made arbitrarily small, and $O_{\epsilon}(f(n))=O(f(n)n^{\epsilon})$.

In this paper, we propose a randomized approximation algorithm for VCP with a tradeoff between the space and the query time. We will show that  for an arbitrary constants $0\leq \beta\leq \frac{2}{3}$ and $0<\delta <1$,   the expected preprocessing time, the expected space, and the query time of our algorithm are $O(n^{4-3\beta}\log n)$, $O(n^{4-3\beta})$,  and $O(\frac{1}{\delta^3}n^{\beta}\log n)$, respectively. The algorithm computes the number of visible segments from $p$, or $m_p$, exactly 
if $m_p\leq \frac{1}{\delta^3}n^{\beta}\log n$. Otherwise,  it computes a $(1+\delta)$-approximation 
$m'_p$   with the probability of at least $1-\frac{1}{\log n}$, where $m_p\leq m'_p\leq (1+\delta)m_p$.
\smallskip
\\
\noindent \textbf{Keywords.} computational geometry, visibility, randomized algorithm, approximation algorithm, graph theory.

\end{abstract}
\section{Introduction}

\subsection*{\it Problem Statement}
Let $S=\{s_1, s_2,\dots, s_n\}$ be a set of $n$ disjoint closed line segments in the plane contained in a bounding box, $\mathbb{B}$. Two points $p$ and $q$ in the bounding box are visible to each other with respect to  $S$, if the open line segment $\overline{pq}$ does not intersect any segments of $S$. A segment $s_i\in S$ is also said to be visible from a point $p$, if there exists a point $q\in s_i$ such that $q$ is visible from $p$. \textit{The visibility counting problem (VCP)} is to find $m_p$, the number of segments of $S$ visible from a query point $p$. We know that \textit{the visibility polygon of a given point $p\in \mathbb{B}$} is defined as
\begin{center}
$VP_{S}(p) = \{ q\in \mathbb{B}: p$ and $q$ are visible$\}$,  
\end{center} and \textit{the visibility polygon of a given segment $s_i$} is defined as
\begin{center}
$VP_{S}(s_i) = \bigcup_{q\in s_i}VP_{S} (q)$.
%=$\{$p\in \mathbb{R}^{2}: s_i$ and $q$ are visible (w.r.t $S$)$\}.$
\end{center}

Consider the $2n$ end-points of the segments of $S$ as vertices of a geometric graph. Add a straight-line-edge between each pair of visible vertices. The result is \textit{the visibility graph of $S$} or $VG(S)$.
We can extend each edge of $VG(S)$ in both directions to the points that the edge hits some segments in $S$ or the bounding box. This creates at most two new vertices and two new edges. Adding all these vertices and edges to $VG(S)$ results in a new geometric graph called \textit{the extended visibility graph of $S$} or $EVG(S)$.
$EVG(S)$ reflects all the visibility information from which the visibility polygon of any segment $s_i\in S$ can be computed \cite{gud}.

\subsection*{\it Related Work}
$VP_S(p)$ can be computed in $O(n\log n)$ time using $O(n)$ space~\cite{asa,sur}. Vegter proposed an output sensitive algorithm that reports $VP_S(p)$ in $O(|VP_S(p)|\log(\frac{n}{|VP_S(p)|}))$ time, by preprocessing the segments in $O(m\log n)$ time using $O(m)$ space, where $m=O(n^{2})$ is the number of edges of $VG(S)$ and $|VP_S(p)|$ is the number of vertices of $VP_S(p)$ \cite{vet}.

$EVG(S)$ can be used to solve VCP. $EVG(S)$ can optimally be computed in $O(n\log n+ m)$ time~\cite{gho1}. If a vertex is assigned to any intersection point of the edges of $EVG(S)$, we have a planar graph, which is called the planar arrangement of the edges of $EVG(S)$.
 All points in any face of this arrangement have the same number of visible segments and this number can be computed for each face in the preprocessing step \cite{gud}. Since there are $O(n^{4})$ faces in the planar arrangement of $EVG(S)$, a point location structure of size $O(n^{4})$ can answer each query in $O(\log n)$ time. But, $O(n^4)$ preprocessing time and space is high. We also know that for any query point $p$, by computing $VP_S(p)$, $m_p$ can be computed in $O(n\log n)$ with no preprocessing. This has led to several results with a tradeoff between the preprocessing cost and the query time~\cite{aro,bos,gho2,poc,zar}.
% (see the visibility book by Ghosh for a complete survey ~\cite{gho2}). 

There are two approximation algorithms for VCP by Fischer \textit{et al.}~\cite{fis1,fis2}.
One of these algorithms uses a data structure of size $O((m/r )^{2})$ to build a $(r/m)$-cutting for $EVG(S)$ by which the queries are answered in $O(\log n)$ time with an absolute error of $r$ compared to the exact answer ($1\leq r\leq n$). The second algorithm uses the random sampling method to build a data structure of size $O((m^{2}\log^{O(1)}n)/l)$ to answer any query in $O(l \log^{O(1)} n)$ time, where $1\leq l\leq n$. In the latter method, the answer of VCP is approximated up to an absolute value of $\delta n$ for any constant $\delta >0$ ($\delta$ affects the constant factor of both data structure size and the query time).

 In \cite{sur}, Suri and O'Rourke represent the visibility polygon of a segment by a union of set of triangles.
Gudmundsson and Morin~\cite{gud} improved the covering scheme of  \cite{sur}.
 Their method builds a data structure of size $O_{\epsilon}(m^{1+\alpha})=O_{\epsilon}(n^{2(1+\alpha)})$ in $O_{\epsilon}(m^{1+\alpha})=O_{\epsilon}(n^{2(1+\alpha)})$ preprocessing time, from which each query is answered in $O_{\epsilon}(m^{(1-\alpha)/2})=O_{\epsilon}(n^{1-\alpha})$ time, where $0<\alpha\leq1$.
This algorithm returns $m'_{p}$ such that $m_{p}\leq m'_{p}\leq 2m_{p}$. The same result can be achieved from \cite{ali} and \cite{nor}.
In \cite{ali}, it is proven that the number of visible end-points of the segments in $S$, denoted by $ve_p$, is a 2-approximation of $m_p$, that is $m_p\leq ve_p\leq 2m_p$. 

\subsection*{\it Our Results}In this paper, we present a randomized ($1+\delta$)-approximation algorithm, where $0<\delta\leq 1$. The expected preprocessing time and space of our algorithm are $O(m^{2-3\beta/2}\log m)$ and $O(m^{2-3\beta/2})$ respectively, and our query time is $O(\frac{1}{\delta^3}m^{\beta/2}\log m)$, where $0\leq \beta\leq \frac{2}{3}$ is chosen arbitrarily in the preprocessing time.

In our proposed algorithm, a graph $G(p)$ is associated to each query point $p$; the construction of $G(p)$ is explained in Section 2. It will be shown that $G(p)$ has a planar embedding and this formula holds: $m_p=n-F(G(p))+1$ or $n-F(G(p))+2$, where $F(G(p))$ is the number of faces of $G(p)$.

Using Euler's formula for planar graphs, we will show that if $p$ is inside a bounded face of $G(p)$, then $m_p=ve_p-C(G(p))+1$, otherwise $m_p=ve_p-C(G(p))$, where $C(G(p))$ is the number of connected components of $G(p)$. In Section 3 and 4, we will present algorithms to approximate $ve_p$ and $C(G(p))$. This leads to an overall approximation for $m_p$.
 
Some detail of our algorithm is as follows: First, we try to calculate $VP_S(p)$ by running the algorithm presented in \cite{vet} for $\frac{1}{\delta^3}m^{\beta/2}\log m$ steps. If this algorithm terminates, the exact value of $m_p$ is calculated, which is obviously less than $\frac{1}{\delta^3}m^{\beta/2}\log m$. Otherwise, our algorithm instead returns $m'_p$, such that $m_p\leq m'_p\leq (1+\delta)m_p$ with the probability of at least $1-\frac{1}{\log n}$. Table~\ref{com} compares the performance of our algorithm with the best known result for this problem. Note that if we choose a constant number $0<\delta<1$, then our query time is better than \cite{gud}, however our algorithm returns a $(1+\delta)$-approximation of the answer with a high probability.

\begin{table}[h]
\caption{Comparison of our method and the best known result for VCP. Note that $\beta$ ($0\leq \beta\leq \frac{2}{3}$) is chosen in the preprocessing time and $1+\delta$ ($0<\delta\leq1$) is the approximation factor of the algorithm which affects the query time and $O_{\epsilon}(f(n))=O(f(n)n^{\epsilon})$, where $\epsilon$ is a constant number that can be arbitrary small.}
\begin{center}
\begin{tabular}{|c|c|c|c|c|c|}
\hline
Reference&Preprocessing time& Space& Query& Approx-Factor \\ \hline
\cite{gud}&$O_{\epsilon}(m^{2-3\beta/2})$&$O_{\epsilon}(m^{2-3\beta/2})$&$O_{\epsilon}(m^{3\beta/4})$&2\\
\hline
Our result&$O(m^{2-3\beta/2}\log m)$&$O(m^{2-3\beta/2})$&$O(\frac{1}{\delta^3}m^{\beta/2}\log m)$&$1+\delta$\\
\hline
\end{tabular}
\end{center}
\label{com}
\end{table}

\begin{figure}[htpb]
	\centering
	\begin{tikzpicture}[scale=0.3]
%.....%-----------1
\draw(-10,10)--(-2,10);
\filldraw(-10,10) circle(2pt);
\filldraw(-2,10) circle(2pt);
\draw (-6,10) node[above] {$s_1$};

\draw (-10,10) node[above]{\tiny{$l(s_1)$}};
\draw (-2,10) node[above]{\tiny{$r(s_1)$}};

\draw(-12,8)--(-9,8.5);
\filldraw(-12,8) circle(2pt);
\filldraw(-9,8.5) circle(2pt);
\draw (-10.5,8.25) node[above] {$s_2$};

\draw (-12,8) node[above]{\tiny{$l(s_2)$}};
\draw (-9,8.5) node[above]{\tiny{$r(s_2)$}};

\draw(-10,6)--(-7,7);
\filldraw(-10,6) circle(2pt);
\filldraw(-7,7) circle(2pt);
\draw (-8.5,6.5) node[above] {$s_3$};

\draw (-10,6) node[above]{\tiny{$l(s_3)$}};
\draw (-7,7) node[above]{\tiny{$r(s_3)$}};

\draw(-5,8)--(-2,7.5);
\filldraw(-5,8) circle(2pt);
\filldraw(-2,7.5) circle(2pt);
\draw (-3.5,7.75) node[above] {$s_4$};

\draw (-5,8) node[above]{\tiny{$l(s_4)$}};
\draw (-2,7.5) node[above]{\tiny{$r(s_4)$}};

\draw(-6,6)--(-1,6.5);
\filldraw(-6,6) circle(2pt);
\filldraw(-1,6.5) circle(2pt);
\draw (-3.5,6.25) node[above] {$s_5$};

\draw (-6,6) node[above]{\tiny{$l(s_1)$}};
\draw (-1,6.5) node[above]{\tiny{$r(s_1)$}};

\filldraw(-5,-2) circle(2pt);
\draw (-5,-2) node[left] {$p$};
\draw (-5,-3) node[below]{(a)};

%\draw(-10,10)--(-5,-2)[dashed];
%\draw(-2,10)--(-5,-2)[dashed];
%\draw(-12,8)--(-5,-2)[dashed];
%\draw(-9,8.5)--(-5,-2)[dashed];
%\draw(-10,6)--(-5,-2)[dashed];
%\draw(-7,7)--(-5,-2)[dashed];
%\draw(-5,8)--(-5,-2)[dashed];
%\draw(-2,7.5)--(-5,-2)[dashed];
%\draw(-6,6)--(-5,-2)[dashed];
%\draw(-1,6.5)--(-5,-2)[dashed];

%2
\draw(4,10)--(12,10);
\filldraw(4,10) circle(2pt);
\filldraw(12,10) circle(2pt);
\draw (8,10) node[above] {$s_1$};

\draw(4,10)--(4.62,8.46);

\draw (4,10) node[above] {\tiny{$a$}};
\draw (4.82,8.46) node[above] {\tiny{$a'$}};

\filldraw(4.62,8.46) circle(2pt);
\draw(12,10)--(11.4,7.6);
\filldraw(11.4,7.6) circle(2pt);

\draw(2,8)--(5,8.5);
\filldraw(2,8) circle(2pt);
\filldraw(5,8.5) circle(2pt);
\draw (3.5,8.25) node[above] {$s_2$};

\draw(5,8.5)--(5.8,6.6);
\filldraw(5.8,6.6) circle(2pt);

\draw(4,6)--(7,7);
\filldraw(4,6) circle(2pt);
\filldraw(7,7) circle(2pt);
\draw (4.7,6.5) node[above] {$s_3$};

\draw(9,8)--(12,7.5);
\filldraw(9,8) circle(2pt);
\filldraw(12,7.5) circle(2pt);
\draw (10.5,7.75) node[above] {$s_4$};

\draw(9,8)--(9,6.1);
\filldraw(9,6.1) circle(2pt);

\draw(12,7.5)--(11.63,6.38);
\filldraw(11.63,6.38) circle(2pt);

\draw(8,6)--(13,6.5);
\filldraw(8,6) circle(2pt);
\filldraw(13,6.5) circle(2pt);
\draw (10.5,6.25) node[above] {$s_5$};

\filldraw(9,-2) circle(2pt);
\draw (9,-2) node[left] {$p$};
\draw (9,-3) node[below]{(b)};

%3
\draw(-10,-10)--(-2,-10)[dashed];
\filldraw(-10,-10) circle(2pt);
\filldraw(-2,-10) circle(2pt);
%\draw (-6,-10) node[above] {\tiny{$s_1$}};

\filldraw(-6,-10.2) circle(2pt);
\draw (-6,-10.2) node[below] {\tiny{$v_1$}};

\draw(-10,-10)--(-6,-10.2)[red];
\draw(-2,-10)--(-6,-10.2);

\draw(-10,-10)--(-9.38,-11.54)[red];
\filldraw(-9.38,-11.54) circle(2pt);
\draw(-2,-10)--(-2.6,-12.4);
\filldraw(-2.6,-12.4) circle(2pt);

\draw(-12,-12)--(-9,-11.5)[dashed];
\filldraw(-12,-12) circle(2pt);
\filldraw(-9,-11.5) circle(2pt);
%\draw (-10.5,-11.75) node[above] {\tiny{$s_2$}};

\filldraw(-10.5,-12) circle(2pt);
\draw (-10.5,-12) node[below] {\tiny{$v_2$}};

\draw(-9,-11.5)--(-10.5,-12);
\draw(-9.38,-11.54)--(-10.5,-12)[red];

\draw(-9,-11.5)--(-8.2,-13.4);
\filldraw(-8.2,-13.4) circle(2pt);

\draw(-10,-14)--(-7,-13)[dashed];
\filldraw(-10,-14) circle(2pt);
\filldraw(-7,-13) circle(2pt);
%\draw (-8.5,-13.5) node[above] {\tiny{$s_3$}};

\filldraw(-8.5,-13.75) circle(2pt);
\draw (-8.5,-13.75) node[below] {\tiny{$v_3$}};

\draw(-8.2,-13.4)--(-8.5,-13.75);

\draw(-5,-12)--(-2,-12.5)[dashed];
\filldraw(-5,-12) circle(2pt);
\filldraw(-2,-12.5) circle(2pt);
%\draw (-3.5,-12.25) node[above] {\tiny{$s_4$}};

\filldraw(-3.5,-12.5) circle(2pt);
\draw (-3.5,-12.5) node[below] {\tiny{$v_4$}};

\draw(-2.6,-12.4)--(-3.5,-12.5);
\draw(-5,-12)--(-3.5,-12.5);
\draw(-2,-12.5)--(-3.5,-12.5);

\draw(-5,-12)--(-5,-13.9);
\filldraw(-5,-13.9) circle(2pt);

\draw(-2,-12.5)--(-2.37,-13.62);
\filldraw(-2.37,-13.62) circle(2pt);

\draw(-6,-14)--(-1,-13.5)[dashed];
\filldraw(-6,-14) circle(2pt);
\filldraw(-1,-13.5) circle(2pt);
%\draw (-3.5,-13.75) node[above] {\tiny{$s_5$}};

\filldraw(-3.5,-14) circle(2pt);
\draw (-3.5,-14) node[below] {\tiny{$v_5$}};

\draw(-5,-13.9)--(-3.5,-14);
\draw(-2.37,-13.62)--(-3.5,-14);

\filldraw(-5,-22) circle(2pt);
\draw (-5,-22) node[left] {$p$};
\draw (-5,-23) node[below]{(c)};

%4----------------

\filldraw(8,-10.2) circle(2pt);
\draw (8,-10.2) node[below] {\tiny{$v_1$}};

\draw (3.5,-12) node[below] {\tiny{$v_2$}};
\filldraw(3.5,-12) circle(2pt);

\draw (5.5,-13.75) node[below] {\tiny{$v_3$}};
\filldraw(5.5,-13.75) circle(2pt);

\draw (10.5,-12.5) node[below] {\tiny{$v_4$}};
\filldraw(10.5,-12.5) circle(2pt);

\draw (10.5,-14) node[below] {\tiny{$v_5$}};
\filldraw(10.5,-14) circle(2pt);

\filldraw(9,-22) circle(2pt);
\draw (9,-22) node[left] {$p$};
\draw (9,-23) node[below]{(d)};

\draw(4,-10)--(8,-10.2);
\draw(12,-10)--(8,-10.2);

\draw(4,-10)--(4.62,-11.54);
%\filldraw(4.62,-11.54) circle(2pt);
\draw(12,-10)--(11.4,-12.4);
%\filldraw(11.4,-12.4) circle(2pt);

\draw(5,-11.5)--(3.5,-12);
\draw(4.62,-11.54)--(3.5,-12);

\draw(5,-11.5)--(5.8,-13.4);
%\filldraw(5.8,-13.4) circle(2pt);

\draw(5.8,-13.4)--(5.5,-13.75);

\draw(11.4,-12.4)--(10.5,-12.5);
\draw(9,-12)--(10.5,-12.5);
\draw(12,-12.5)--(10.5,-12.5);

\draw(9,-12)--(9,-13.9);
%\filldraw(9,-13.9) circle(2pt);

\draw(12,-12.5)--(11.63,-13.62);
%\filldraw(11.63,-13.62) circle(2pt);

\draw(9,-13.9)--(10.5,-14);
\draw(11.63,-13.62)--(10.5,-14);

	\end{tikzpicture}
\caption{
The steps to draw a planar embedding of $G(p)$. (a) The segments are $s_1,\dots, s_5$ with their left and right end-points and a given query point is $p$. (b) For each end-point $a\in s_i$ not visible to $p$, if $a'\in s_j$ such that $pr(a')=a$, we draw $\overline{aa'}$. (c) Put a vertex $v_i$ for each segment $s_i$ in a distance sufficiently close to the middle of $s_i$. For each $a$ and $a'$ (described in (b)), connect  $a$ to $v_i$ and $a'$ to $v_j$. This creats an edge between $v_i$ and $v_j$ shown in red (d) Remove the segments and the remaining is the planar embedding of $G(p)$. Note that the final embedding has $5$ vertices and $5$ edges and each edge is drown as $3$ consequence straight lines.}
\label{drawg}
\end{figure}
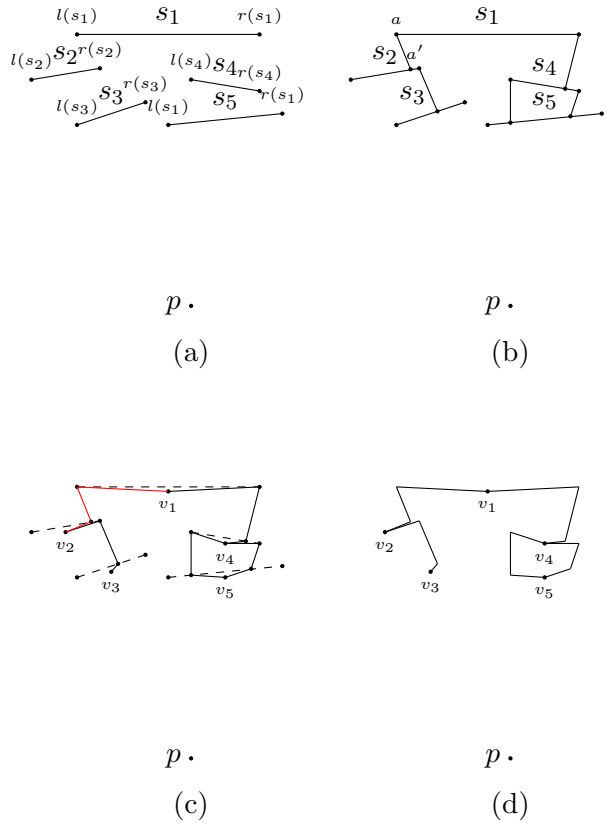

\section{Definitions and the main theorem}
For each point $a'\in s_i$, let $\overrightarrow{pa'}$ be the ray emanating from  the query point $p$ toward $a'$ and let $a=pr(a')$ be the first intersection point of $\overrightarrow{pa'}$ and a segment in $S$ or the bounding box right after touching $a'$. We say that $a=pr(a')$ is covered by $a'$ or the projection of $a'$ is $a$.
Also, suppose that $\overline{x'y'}$ is a subsegment of $s_i$ and $\overline{xy}$ is a subsegment of $s_j$, such that $pr(x')=x$ and $pr(y')=y$ and for any point $z'\in \overline{x'y'}$, $pr(z')\in \overline{xy}$, then we say that $\overline{xy}$ is covered by $\overline{x'y'}$.

For each query point $p$, we construct a graph denoted by $G(p)$ as follows: a vertex $v_i$ is associated to each  segment $s_i\in S$, and an edge $(v_i,v_j)$ is put if $s_j$  covers one end-point of $s_i$ (or vice-versa; that is, if $s_i$ covers one end-point of $s_j$). Obviously, there are two edges between $v_i$ and $v_j$, if $s_j$ (or $s_i$) covers both end-points of $s_i$ (or $s_j$). As an example, refer to Fig~\ref{drawg}.(a) and (d). Note that the bounding box is not considered here.

For any segment $s\in S$, let $l(s)$ and $r(s)$ be the first and second end-points of $s$, respectively swept by a ray around $p$ in clockwise order (Fig~\ref{drawg}.(a)).
%There are one or two edges between two vertices of $G(p)$ according to the property that one or both end-points of one of the segments associated to a vertex is covered by another segment. Notice that it is not possible that an end-point of each of these two segments is covered by the other segment. For an example of $G(p)$ refer to Fig~\ref{drawg}.

%\begin{figure}[htb]
%\begin{center}
%\includegraphics[width=0.6\textwidth]{a1.pdf}
%\caption{The graph of a given query point $p$, $G(p$).}
%\label{gp}
%\end{center}
%\end{figure}
\begin{lemma}
\label{planargp}
$G(p)$ has a planar embedding.
\end{lemma}
\begin{proof}
Here is the construction. For each end-point $a\in s_i$ not visible from $p$, let $a'\in s_j$ such that $pr(a')=a$. Draw the straight-line ${\overline{aa'}}$.
%For each end-point $a$ of a segment in $S$, choose the unique point $b$ on some other segment such that $pr(b)=a$ and $a$ is not on the bounding box. Draw the segment $\overline{ba}$
Doing this, we have a collection of non-intersecting straight-lines. For each $s_i$, we put a vertex $v_i$ located very close to the mid-point of $s_i$. Also, for each segment $\overline{aa'}$, we connect $a$ to $v_i$ and $a'$ to $v_j$. This creates an edge consisting of three consecutive straight-lines $\overline{v_ia}$, $\overline{aa'}$, and $\overline{a'v_j}$  that connects $v_i$ to $v_j$. Obviously, none of these edges intersect. Finally, all the original segments are removed. The remaining is the vertices and edges of a planar embedding of $G(p)$ (These steps can be seen in Fig~\ref{drawg}).

%Imagine the rays emanating from the end-points of segments toward $p$. These rays do not intersect each other except in $p$ (See Fig~\ref{drawg}.a).
%For each end-point $a\in s_i$ which is not visible from $p$, draw $r_{\overrightarrow{ap}}$ until it intersects the first segment $s_j$ at a point $a'\in s_j$ (See Fig~\ref{drawg}.b). Obviously, $\overline{aa'}$ does not intersect any segment of $S$ except $s_j$ and also it does not intersect  any ray emanating from the end-points toward $p$. In the next step, for any segment $s_i\in S$, put a vertex $v_i$ with a distance sufficiently small from the middle of $s_i$, then connect the intersection points of the rays and $s_i$ to $v_i$. Furthermore, connect each of the end-points of $s_i$ which are not visible to $p$ to $v_i$(See Fig~\ref{drawg}.c) Then remove the segments. Thus, we have a geometric planar embedding of $G(p)$ (See Fig~\ref{drawg}.d).

\end{proof}
%\begin{figure}[htb]
%\begin{center}
%\includegraphics[width=0.7\textwidth]{drawg1.pdf}
%\caption{The steps to draw a planar embedding of $G(p)$. a. Consider the rays emanating from the end-points toward $p$. b. Draw each ray emanating from an end-point that is not visible to $p$ until it intersects the first segment. c. Put a vertex for each segment in a distance sufficiently small from the middle of that segment and connect each of these rays to the related vertices. d. The planar embeding of $G(p)$. }
%\label{drawg}
%\end{center}
%\end{figure}
From now on, we use $G(p)$ as the planar embedding of the graph $G(p)$.
As we know the Euler's formula for any non-connected planar graph $G$ with multiple edges is:
\begin{center}
\label{eul}
$V(G)-E(G)+F(G)=1+C(G)$,
\end{center}
where $E(G)$, $V(G)$, $F(G)$, and $C(G)$ are the number of edges, vertices, faces, and connected components of $G$, respectively.
The following theorem provides a method to calculate $m_p$, using $G(p)$.
\begin{theorem}
\label{f}
The number of segments not visible from $p$ is equal to $F(G(p))-2$ if $p$ is inside a bounded face of $G(p)$, or is equal to $F(G(p))-1$, otherwise.
\end{theorem}
%\begin{figure}
%\begin{center}
%\includegraphics[width=0.5\textwidth]{projection.pdf}
%\caption{The projection of  $S'=\{s'_{1},s'_{2},s'_{3},s'_{4},s'_{3}\}$ consecutively covers $s_i$ from $l(s_i)$ to $r(s_i)$.}
%\label{g}
%\end{center}
%\end{figure}

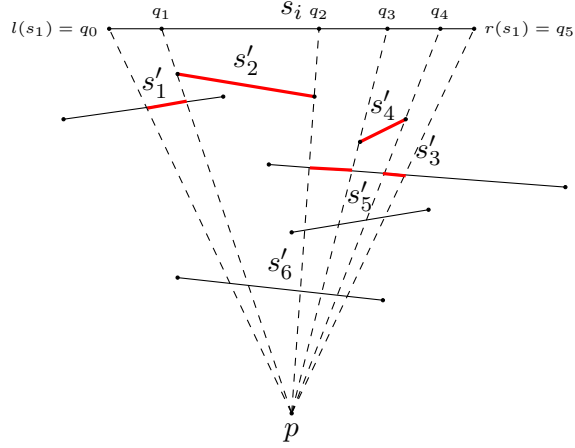
\begin{figure}[htpb]
	\centering
	\begin{tikzpicture}[scale=0.3]
%.....%-----------1
\draw(-8,10)--(8,10);
\filldraw(-8,10) circle(2pt);
\filldraw(8,10) circle(2pt);
\draw (0,10) node[above] {$s_i$};

\draw (-8,10) node[left]{\tiny{$l(s_1)=q_0$}};
\draw (8,10) node[right]{\tiny{$r(s_1)=q_5$}};

\draw(-10,6)--(-3,7);
\draw[very thick,red] (-6.33,6.49)--(-4.6,6.8);

\filldraw(-10,6) circle(2pt);
\filldraw(-3,7) circle(2pt);
\draw (-6,6.5) node[above] {$s'_1$};

\draw(0,-7)--(-8,10)[dashed];

\draw[very thick,red](-5,8)--(1,7);
\filldraw(-5,8) circle(2pt);
\filldraw(1,7) circle(2pt);
\draw (-2,7.5) node[above] {$s'_2$};

\draw(0,-7)--(-5.7,10)[dashed];
\filldraw(-5.7,10) circle(2pt);
\draw (-5.7,10) node[above]{\tiny{$q_1$}};

\draw(0,-7)--(1.2,10)[dashed];
\filldraw(1.2,10) circle(2pt);
\draw (1.2,10) node[above]{\tiny{$q_2$}};

\draw(0,-7)--(4.2,10)[dashed];
\filldraw(4.2,10) circle(2pt);
\draw (4.2,10) node[above]{\tiny{$q_3$}};

\draw(0,-7)--(6.52,10)[dashed];
\filldraw(6.52,10) circle(2pt);
\draw (6.52,10) node[above]{\tiny{$q_4$}};

\draw(0,-7)--(8,10)[dashed];
%\filldraw(8,10) circle(2pt);
%\draw (8,10) node[above]{\tiny{$q_5$}};

\draw(-1,4)--(12,3);
\draw[very thick,red](0.8,3.85)--(2.65,3.75);
\draw[very thick,red](4,3.6)--(5,3.5);
\filldraw(-1,4) circle(2pt);
\filldraw(12,3) circle(2pt);
\draw (6,3.5) node[above] {$s'_3$};

\draw[very thick,red](3,5)--(5,6);
\filldraw(3,5) circle(2pt);
\filldraw(5,6) circle(2pt);
\draw (4,5.5) node[above] {$s'_4$};

\draw(0,1)--(6,2);
\filldraw(0,1) circle(2pt);
\filldraw(6,2) circle(2pt);
\draw (3,1.5) node[above] {$s'_5$};

\draw(-5,-1)--(4,-2);
\filldraw(-5,-1) circle(2pt);
\filldraw(4,-2) circle(2pt);
\draw (-0.5,-1.5) node[above] {$s'_6$};

\filldraw (0,-7) circle(2pt);
\draw (0,-7) node[below]{$p$};

\end{tikzpicture}
\caption{$s_i$ is not visible from $p$. It can be partitioned into 5 subsegments $\overline{q_0q_1}, \overline{q_1q_2}, \overline{q_2q_3}, \overline{q_3q_4}$, and $\overline{q_4q_5}$, each is covered respectively by  subsegment of $s'_1,s'_2,s'_3,s'_4$, and $s'_3$ shown above.}
\label{g}
\end{figure}

\begin{proof}
We construct a bijection $\phi$ between the segments not visible from $p$ to the faces of $G(p)$ except the unbounded face and the face that contains $p$. This will compelete the proof of our theorem.

Suppose that $s_i$ is a segment not visible from $p$. Then, we can partition $s_i$ into $k$ subsegments, $\overline{q_0q_1}, \overline{q_1q_2}, \dots, \overline{q_{k-1}q_k} $ such that $q_0=l(s_i)$, $q_k=r(s_i)$, and for each $\overline{q_iq_{i+1}}$, there is a subsegment $\overline{q'_iq'_{i+1}}\in s_j$ that covers $\overline{q_iq_{i+1}}$. 
Let $s'_{1}, s'_{2},\dots, s'_{k}$ be the set of segments such that  $\overline{xy}\in s'_{i+1}$ covers $\overline{q_iq_{i+1}}$ (note that some segments may appear more than once in the above sequence) (Fig~\ref{g}).
We claim that the vertices $v_i, v'_{1}, v'_{2},\dots, v'_{k}$ form a bounded face of $G(p)$ that does not contain $p$. In $\phi$, we associate this face to $s_i$. Since $v'_1$ is the vertex associated to the first segment that covers $\overline{q_0q_1}$, $s'_1$ will cover $l(s_i)$ and hence $v_i$ is adjacent to $v'_1$. Similarly, since $s'_k$ covers $r(s_i)$, hence $v_i$ is adjacent to $v'_k$.
The next subsegment that covers a subsegment of $s_i$ comes from $s'_2$. This means that $r(s'_1)$ is covered by $s'_2$ or $l(s'_2)$ is covered by $s'_1$. This implies that $v'_1$ is adjacent to $v'_2$. Similarly, we can show that $v'_i$ is adjacent to $v'_{i+1}$ for all $1\leq i<k$. To complete the construction, we need to show that the closed path formed by $v_i\rightarrow v'_{1}\rightarrow v'_{2},\dots \rightarrow v'_{k}\rightarrow v_i$ is a bounded face not containing $p$. Consider a ray around $p$ in clockwise order. 
The area that this ray touches under $s_i$ and above $s'_1,\dots, s'_k$ is a region bounded by $v_i, v'_{1}, v'_{2},\dots,v'_{k}$. Obviously, $p$ is not inside this region.

Now, we show that our map $\phi$ is one-to-one and onto. The proof of one-to-oneness is easier. If $\phi(s_i)= \phi(s_j)$, then according to the construction of $\phi$, a subsegment of $s_i$ covers a subsegment of $s_j$ and a subsegment of $s_j$ covers a subsegment of $s_i$. This is a contradiction since these segments do not intersect. To prove the onto-ness, we need to show for any bounded face $f$ that does not contain $p$, there is a vertex $v_i$ corresponding to a segment $s_i$ that is not visible to $p$ such that $\phi(s_i)=f$.

To find $s_i$, we use the sweeping ray around $p$. Since $f$ is assumed to be bounded and not containing $p$, the face $f$ is between two rays from $p$; one from the left and the other from the right. If we start sweeping from left to right, there is a segment corresponding to the vertices of $f$ whose end-point is the first to be covered by the other segments corresponding to the vertices of $f$. We claim that $s_i$ is the desired segment .i.e. $s_i$ is not visible to $p$ and $\phi(s_i)=f$. For example in Fig~\ref{g}, the closed path $v_i\rightarrow v'_1\rightarrow v'_2\rightarrow v'_3\rightarrow v'_4,\rightarrow v'_3, \rightarrow v_i$ forms a face and $s_i$ is the first segment among $\{s_i, s'_1, s'_2, s'_3, s'_4\}$ such that $l(s_i)$ is covered by one of the segments in $\{s_i, s'_1, s'_2, s'_3, s'_4\}$.

Obviously, $l(s_i)$ is not visible from $p$. $v'_{1}$ is adjacent to $v_i$ which means that a subsegment of $s'_{1}$ covers a subsegment of $s_i$. Since $v'_{1}$ and $v'_{2}$ are adjacent, this means that a subsegment of $s'_{2}$ consecutively covers the next subsegment of $s_i$ right after $s'_{1}$. Continuing this procedure, we conclude that a subsegment of each $s'_{i}$ covers some subsegment of $s_i$ continuously right after $s'_{i-1}$. $v'_{k}$ and $v_i$ are also adjacent, so $r(s_i)$ is not visible from $p$.
We conclude that subsegments of $s'_1, s'_2 \dots, s_k$ completely cover $s_i$ and hence $s_i$ is not visible from $p$.

So, if $p$ is in the unbounded face of $G(p)$, the number of segments which are not visible from $p$ is $F(G(p))-1$, otherwise it is $F(G(p))-2$.
\end{proof} 
The Euler's formula is used to compute $F(G(p))$.
Obviously, $V(G(p))$ is $n$. For each end-point not visible from $p$, an edge is added to $G(p)$; therefore, $E(G(p))$ is $2n-ve_p$ ($ve_p$ was defined above as the number of visible end-points from $p$). The Euler's formula and Theorem~\ref{f} indicate the following lemma.  
 \begin{lemma}
 \label{o2}
 If $p$ is inside a bounded face of $G(p)$, then
 $m_p=ve_p-C(G(p))+1$, otherwise,
 $m_p=ve_p-C(G(p))$. 
 \end{lemma}
 
In the rest of this paper, two algorithms are presented; one to approximate $ve_p$ and the other to approximate $C(G(p))$. By applying Lemma~\ref{o2}, an approximation value of $m_p$ is calculated. The main result of this paper is thus derived from the following theorem. The proof is given in the Appendix. 
\begin{theorem}(Main theorem)
\label{maint}
 For any $0<\delta\leq1$ and $0\leq \beta \leq \frac{2}{3}$, VCP can be approximated in $O(\frac{1}{\delta^3}m^{\beta/2}\log m)$ query time using $O(m^{2-3\beta/2}\log m)$ expected preprocessing time and $O(m^{2-3\beta/2})$ expected space. This algorithm returns a value $m_p'$ such that with the probability at least $1-\frac{1}{\log m}$, $m_p\leq m_p'\leq (1+\delta) m_p$ when $m_p\geq \frac{1}{\delta^3}m^{\beta/2}\log m$ and returns the exact value when $m_p< \frac{1}{\delta^3}m^{\beta/2}\log m$.
\end{theorem}

%\section{Proof of Lemma \ref{approx-end-point}}
\section{An approximation algorithm to compute the number of visible end-points}
In this section, we present an algorithm to approximate $ve_p$, the number of visible end-points. In the preprocessing phase, we build the data structure of the algorithm presented in \cite{vet} which calculates $VP_S(p)$ in $O(|VP_S(p)|\log(n/|VP_S(p)|))$ time, where $|VP_S(p)|$ is the number of vertices of $VP_S(p)$.
In \cite{vet}, the algorithm for computing $VP_S(p)$, consists of a rotational sweep of a line around $p$. During the sweep, the subsegments visible from $p$ along the sweep-line are collected.
In the preprocessing phase, we choose a fixed parameter $\beta$, where $0\leq \beta \leq \frac{2}{3}$. In the query time we also choose a fixed parameter $0<\delta \leq 1$ which is the value of approximation factor of the algorithm.

We use the algorithm presented in \cite{vet} to find the visible end-points, but for any query point, we stop the algorithm if more than $ \frac{2}{\delta^3}m^{\beta/2}\log m$ of the visible end-points are found. 

If  the sweep line completely sweeps around $p$ before counting $\frac{1}{\delta^3}m^{\beta/2}\log m$ of the visible end-points, then we have completely computed $VP_S(p)$ and we have $|VP_S(p)|\leq  \frac{2}{\delta^3}m^{\beta/2}\log m$.
In this case, the number of visible segments can be calculated exactly in $O( \frac{1}{\delta^3}m^{\beta/2}\log m)$ time. Otherwise, $ve_p> \frac{2}{\delta^3}m^{\beta/2}\log m$ and the answer is calculated in the next step of algorithm, that we now explain.

The visibility polygon of an end-point $a$ is a star shaped polygon consisting of $m_a=O(n)$ non-overlapping triangles \cite{asa,sur}, which are called \textit{the visibility triangles of $a$} denoted by $VT_S(a)$. Notice that $m_a$ is the number of edges of $EVG(S)$ incident to $a$.
The query point $p$ is visible to an end-point $a$, if and only if it lies inside one of the visibility triangles of $a$. Let $VT_S$ be the set of  visibility triangles of all the end-points of the segments in $S$. Then, the number of visible end-points from $p$ is the number of triangles in $VT_S$ containing $p$.  We can construct $VT_S$ in $O(m\log m)=O(n^2\log n)$ time using $EVG(S)$ and $|VT_S|=O(m)=O(n^2)$\cite{gud}.

We can preprocess a given set of triangles using the following lemma to count the number of triangles containing any query point.
\begin{lemma}
\label{arr}
Let $\Delta$ be a set of $n$ triangles. There exists a data structure of size $O(n^2)$, such that in the preprocessing time of $O(n^2\log n)$,  the number of triangles containing a query point $p$ can be calculated in $O(\log n)$ time.
\end{lemma}
\begin{proof}
Consider the planar arrangement of the edges of the triangles in $\Delta$ as a planar graph. 
Let $f$ be a face of this graph. Then, for any pair of  points $p$ and $q$ in $f$, the number of triangles containing $p$ and $q$ are equal. Therefore, we can compute these numbers for each face in a preprocessing phase and then, for any query point locate the face containing that point. There are $O(n^{2})$ faces in the planar arrangement of $\Delta$, so a point location structure of size $O(n^{2})$ can answer each query in $O(\log n)$ time\cite{kirk}. 
Note that the number of triangles containing a query point differs in $1$ for  any pair of adjacent faces.
\end{proof}

\subsection{The algorithm}
Here, we present an algorithm to approximate $ve_p$. We use this algorithm when $m_p>  \frac{1}{\delta^3}m^{\beta/2}\log m$.
In the preprocessing phase we take a random subset $RVT_1\subset VT_S$ such that each member of $VT_S$ is chosen with the probability of $\frac{1}{m^{\beta}}$.
\begin{lemma}
\label{exp}
$E(|RVT_1|)=O(m^{1-\beta})$.
\end{lemma}
\begin{proof}
Let $VT_S=\{\Delta_1,\Delta_2,\dots,\Delta_{m'}\}$, where $m'=O(m)=O(n^2)$ and $X_i=1$ if $\Delta_i\in RTV_1$, and $X_i=0$ otherwise. We have,
\begin{center}
 $E(|RVT_1|)=E(\sum^{m'}_{i=1} X_i)=\sum^{m'}_{i=1} E(X_i)=\sum^{m'}_{i=1} \frac{1}{m^{\beta}}=\frac{m'}{m^{\beta}}=O(m^{1-\beta})$.
 \end{center}
\end{proof}
Suppose that in the preprocessing time, we choose $m^{\beta/2}$ independent random subsets  $RVT_1,\dots, RVT_{m^{\beta/2}}$ of $VT_S$. Using Lemma~\ref{arr}, for any query point  $p$, the number of triangles of each $RVT_i$ containing $p$ denoted by $(ve_p)_i$, is calculated in $O(\log m)$ time by $O(m^{2-2\beta}\log m)$ expected preprocessing time and $O(m^{2-2\beta})$ expected space. Then, $ve'_p=m^{\beta}\frac {\sum^{m^{\beta/2}}_{i=1} (ve_p)_i} {m^{\beta/2}}$ is returned as the approximation value of $ve_p$.
%We show that with a high probability close to 1, $ve'_p$ is a good approximation of $ve_p$.

\subsection{Analysis of approximation factor}
In this section the approximation factor of the algorithm is calculated. Let $X_i=m^{\beta}(ve_p)_i$. 
\begin{lemma}
$E(X_i)=ve_p$.
\end{lemma}
\begin{proof}
Suppose that $VT(p)=\{\Delta'_{1},\Delta'_{2},\dots, \Delta'_{{ve_{p}}}\}\subset VT_S$ be the set of all triangles containing $p$. Let $Y_j=1$ if $\Delta'_{j}\in RVT_i$, and $Y_j=0$ otherwise. So, $(ve_p)_i=\sum^{ve_{p}}_{j=1} Y_j$ and $E((ve_p)_i)=E(\sum^{ve_{p}}_{j=1} Y_j)=\frac{ve_p}{m^{\beta}}$.
$E(X_i)=E(m^{\beta}(ve_p)_i)=m^{\beta}E((ve_p)_i)=m^{\beta}\frac{ve_p}{m^{\beta}}=ve_p$.
\end{proof}
In addition, we can conclude the following lemma:
\begin{lemma}
$E(\frac{\sum^{m^{\beta/2}}_{i=1} X_i} {m^{\beta/2}})=ve_p$.
\end{lemma}
So, $X_1,X_2,\dots, X_{m^{\beta/2}}$ are random variables with $E(X_i)=ve_p$. According to Chebyshev's Lemma the following lemma holds
\begin{lemma}
(Chebyshev's Lemma)
\label{cheb}
Given $X_1, X_2, \dots, X_n$ sequence of i.i.d.'s random variables with finite expected value $E(X_1) = E(X_2) = \dots = \mu$, we have,
\begin{center} 
$P((|\frac{X_1+\dots +X_n}{n}-\mu|)>\varepsilon_1 )\leq \frac{Var(X)}{n{\varepsilon_1}^2}$.
\end{center}
\end{lemma}

\begin{lemma}
With a probability at least $1-\frac{1}{\log m}$ we have,
\begin{center}
$(1-\delta)ve_p\leq ve'_p \leq(1+\delta)ve_p$.
\end{center}
\end{lemma}
\begin{proof}
Using Lemma~\ref{cheb},  we choose $\varepsilon_1=\delta ve_p$. Here, $\delta$ indicates the approximation factor of the algorithm. Obviously, $Var(X_i)=m^{2\beta}(ve_p)(1-\frac{1}{m^{\beta}})\frac{1}{m^{\beta}}$. So,
\begin{center}
 $\mathbb{P}=P(|ve'_p-ve_p|>\delta ve_p)\leq \frac {m^{\beta} ve_p}{m^{\beta/2}{\delta}^2 (ve_p)^2}$.
\end{center} 
We know that $ve_p\geq \frac{1}{\delta^2}m^{\beta/2}\log m$, so
\begin{center}
 $\mathbb{P}=P(|ve'_p-ve_p|>\delta ve_p)\leq \frac {1}{\log m}$.
\end{center} 
With the probability of at least $1-\mathbb{P}$, we have,
\begin{center}
$(1-\delta)ve_p\leq ve'_p \leq(1+\delta)ve_p$.
\end{center}
Also, for a large $m$, we have $\mathbb{P}\sim 0$.
\end{proof}

\subsection{Analysis of time and space complexity}

In the first step of the query time, we run the algorithm of \cite{vet}. The preprocessing time and space for constructing the data structure of \cite{vet} are $O(m\log m)$ and $O(m)$, respectively, which computes $VP_S(p)$ in $O(|VP_S(p)|\log(n/|VP_S(p)|))$ time. As we run this algorithm for at most $\frac{1}{\delta^3} m^{\beta/2}\log m$ steps, the query time of the first step is $O(\frac{1}{\delta^3} m^{\beta/2}\log m)$.

According to Lemma~\ref{exp}, $E(|RVT_i|)=O(m^{1-\beta})$.
Using Lemma~\ref{arr}, the expected preprocessing time and space for each $RVT_i$ are $O(m^{2-2\beta}\log m)$ and $O(m^{2-2\beta})$ respectively, such that in $O(\log m)$ we can calculate $(ve_p)_i$.
So, the expected preprocessing time and space are $m^{\beta/2}O(m^{2-2\beta}\log m)=O(m^{2-\frac{3}{2}\beta}\log m)$ and $m^{\beta/2}O(m^{2-2\beta})=O(m^{2-\frac{3}{2}\beta})$ respectively.

In the second step, for each $RVT_i$ the value of $(ve_p)_i$ is calculated in $O(\log m)$. Therefore, the query time is $O(\frac{1}{\delta^3}m^{\beta/2}\log m)+O(m^{\beta/2}\log m)$.
So, we have the following lemma.
 \begin{lemma}
\label{approx-end-point}
There exists an algorithm that for any query point $p$, approximates $ve_p$ in $O(\frac{1}{\delta^3}m^{\beta/2}\log m)$ query time using $O(m^{2-3\beta/2}\log m)$ expected preprocessing time and $O(m^{2-3\beta/2})$ expected space $(0\leq \beta\leq \frac{2}{3})$. This algorithm returns the exact value of $ve_p$ when $ve_p<\frac{1}{\delta^2}m^{\beta/2}\log m$. Otherwise,  a value of $ve_p'$ is returned such that with the probability of at least $1-\frac{1}{\log m}$, we have $(1-\delta)ve_p\leq ve_p'\leq (1+\delta)ve_p$.
\end{lemma}

%\section{Proof of Lemma \ref{approx-components} and Theorem~\ref{maint}}
\section{An approximation algorithm for computing the number of components of $G(p)$}

In this section, we explain an algorithm to compute the number of connected components of $G(p)$, each is simply called a component of $G(p)$.

Let $c$ be a component such that  $p$ is not inside any of its faces. 
Without loss of generality we can assume that $p$ lies below $c$.
It is easy to see that there exist rays emanating from $p$ that do not intersect any segments corresponding to the vertices of $c$. We start sweeping one of these rays in a clockwise direction. Let $l(c)$ (left end-point of $c$) be the first end-point of a segment of $c$ and $r(c)$ (right end-point of $c$) be the last end-point of a segment of $c$ that are crossed by this ray. (Fig~\ref{vispart}). This way every component $c$ has $l(c)$ and $r(c)$ except the component containing $p$. Also, note that $r(c)$ and $l(c)$ do not depend on the choice of the starting ray. 
As said,  the bounding box is not a part of $G(p)$, but $G(p)$ is contained in the bounding box. %\begin{figure}[htb]
%\begin{center}
%\includegraphics[width=0.7\textwidth]{vispart2.pdf}
%\caption{The visible parts of a segment and the right and left end-points of a component are shown. $G(p)$ has three components; $\{s_1,s_6\}$, $\{s_3,s_4\}$ and $\{s_5\}$.}
%\label{vispart}
%\end{center}
%\end{figure}

\begin{figure}[htpb]
	\centering
	\begin{tikzpicture}[scale=0.3]
	
%.....%-----------1
\draw(0,20)--(20,20);
\draw(0,0)--(0,20);
\draw(0,0)--(20,0);
\draw(20,0)--(20,20);

\draw(9.5,17)node[above]{$s_1$};
\draw(2,18)--(17,16);
\filldraw(2,18)circle(2pt);
\filldraw(17,16)circle(2pt);

\draw(17,14.5)node[above]{$s_6$};
\draw(16,15)--(18,14);
\filldraw(16,15)circle(2pt);
\filldraw(18,14)circle(2pt);

\draw(18,14)node[right]{\tiny $l(s_6)$};

\draw(3,15.25)node[above]{$s_2$};
\draw(2,15)--(4,15.5);
\filldraw(2,15)circle(2pt);
\draw(2,15)node[left]{\tiny $l(s_2)$};

\filldraw(4,15.5)circle(2pt);

\draw(9,12.5)node[above]{$s_3$};
\draw(7,12)--(11,13);
\filldraw(7,12)circle(2pt);
\filldraw(11,13)circle(2pt);
\draw(7,12)node[left]{\tiny $l(s_3)$};

\draw(11.5,12.1)node[above]{$s_4$};
\draw(10,12.2)--(13,12);
\filldraw(10,12.2)circle(2pt);
\filldraw(13,12)circle(2pt);

\draw(13,12)node[right]{\tiny $r(s_4)$};

\draw(11,11.1)node[above]{$s_5$};
\draw(10.5,11.2)--(11.5,11);
\filldraw(10.5,11.2)circle(2pt);
\filldraw(11.5,11)circle(2pt);
\draw(10.5,11.2)node[left]{\tiny $l(s_5)$};

\draw(11.5,11)node[right]{\tiny $r(s_5)$};

\filldraw(9,2)circle(2pt);
\draw(9,2)node[below]{$p$};

\draw(9,2)--(0,18.6)[dashed];
\draw(9,2)--(3.3,17.7)[dashed];
\draw[thick](5.9,17.5)--(3.3,17.8);
\filldraw(5.9,17.5)circle(2pt);
\filldraw(3.3,17.8)circle(2pt);
\draw(5.9,17.5)node[above]{\tiny{$a'$}};
\draw(3.3,17.8)node[above]{\tiny{$a$}};

\draw(9,2)--(16.5,16.1)[dashed];
\filldraw(16.5,16.1)circle(2pt);
\draw(16.5,16.1)node[above]{\tiny{$b'$}};

\filldraw(14.8,16.3)circle(2pt);
\draw(14.8,16.3)node[above]{\tiny{$b$}};
\draw[thick](16.5,16.1)--(14.8,16.3);

\draw(9,2)--(20,16.6)[dashed];
\draw(9,2)--(5.9,17.5)[dashed];
\draw(9,2)--(14.8,16.3)[dashed];
\draw(9,2)--(11.9,12.1)[dashed];
\draw(9,2)--(10.7,12.3)[dashed];

\filldraw(20,16.6)circle(2pt);
\draw(20,16.6)node[right]{\tiny{$c$}};

\filldraw(0,18.6)circle(2pt);
\draw(0,18.6)node[left]{\tiny{$c'$}};

\end{tikzpicture}
\caption{$\overline{aa'}$ and $\overline{bb'}$ are the visible subsegments of $s_1$. The bounding box has one visible part from $c$ to $c'$. $G(p)$ has three components; $\{s_1,s_2,s_6\}$, $\{s_3,s_4\}$, and $\{s_5\}$. $l(s_2)$, $l(s_3)$, and $l(s_5)$ are the left end-points of these components, respectively. $r(s_6)$, $r(s_4)$, and $r(s_5)$ are the right end-points of these components, respectively.}
\label{vispart}
\end{figure}
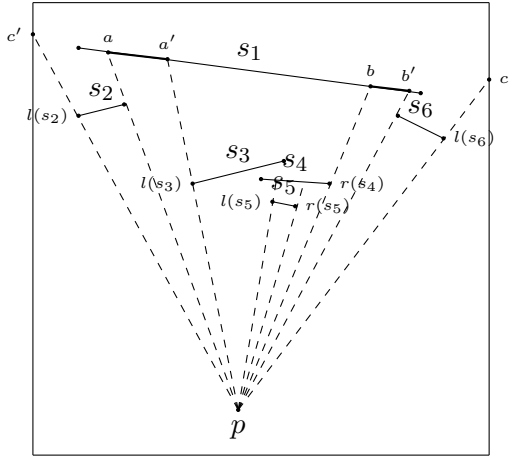

\begin{lemma}
\label{seg}
For each component $c$, except the one containing $p$, the projections of $l(c)$ and $r(c)$ both belong either to the same segment or the bounding box.
\end{lemma}
\begin{proof}
Assume that $pr(l(c))$ belongs to a segment $s\in S$. Since $l(s)$ is on the left of $l(c)$, $s$ can not be among the segments of $c$. We claim that $r(s)$ is on the right of $r(c)$. Obviously, if this claim is true then, if $pr(r(c))\in s'$, then $l(s')$ is on the left of $l(c)$. Clearly, if $s\neq s'$, then these two should intersect, which is impossible. Also, this implies that if $pr(l(c))$ is on the bounding box, then $pr(r(c))$ should to be on the bounding box as well.
The claim is proven by contradiction. Assume that $r(s)$ is on the left of $r(c)$. Since, $r(s)$ is not visible from $p$, then there should exist a segment $s'$ that covers $r(s)$. Since, $s$ is not in $c$ and $s'$ is connected to $s$, $s'$ can not be in $c$, so $l(s')$ is to the right of $l(c)$ and hence is not visible. Therefore, there should exist a different segment $s''$ that covers $l(c)$ and with the same argument $s''$ can not be in $c$ and $l(s'')$ should be covered by another segment. This process can not be continued indefinitely since the number of segments is finite and therefore we will reach a contradiction.
\end{proof}

Let $s'_1, s'_2, s'_3$, and $s'_4$ be the segments of the bounding box.
According to Lemma~\ref{seg}, we can associate a pair of adjacent visible subsegments or a connected visible part of the bounding box for each component of $G(p)$. For example, in Fig~\ref{vispart}, $s_1$ has two visible subsegments which are associated to the component  composed of $s_3$ and $s_4$. If we can count the number of visible subsegments of each segment and the number of visible parts of the bounding box, then we can compute the exact value of $C(G(p))$. Because each pair of consecutive visible subsegments of a segment and each visible part of the bounding box are associated to a component. 
Let $c'$ be the number of visible parts of the bounding box. If $c'>0$, then $p$ is in the unbounded face. So, if each segment $s_i$ has $c_i$ visible subsegments, then $C(G(p))=c'+\sum^n_{i=1} \max{\{(c_i-1),0\}}$ . For example in Fig~\ref{vispart}, $c_1=2$, $c_2=1$, $c_3=1$, $c_4=2$,  $c_5=1$ and $c_6=1$, also $c'=1$. This implied that $C(G(p))=3$.
If $c'=0$, then $p$ is in a bounded face and this face is contained in a component with no left and right end-point, so in this case $C(G(p))=1+\sum^n_{i=1} \max{\{(c_i-1),0\}}$.
 
In the following we propose an algorithm to approximate the number of visible subsegments of each segment $s_i\in S\cup \{s'_1,s'_2,s'_3,s'_4\}$.
\subsection{Algorithm}
According to \cite{gud}, it is possible to cover the visibility region of each segment $s_i\in S\cup \{s'_1,s'_2,s'_3,s'_4\}$ with $O(m_{s_i})$ triangles denoted by $VT(s_i)$. Here, $|VT(s_i)|=O(m_{s_i})$, where  $m_{s_i}$ is the number of edges of $EVG(S)$ incident on $s_i$. Note that the visibility triangles of $s_i$ may overlap. If we consider the visibility triangles of all segments, then there is a set $VT_S=\{\Delta_1,\Delta_2,\dots\}$ of $|VT_S|=O(m)$ triangles.
We say $\Delta_i$ is related to $s_j$ if and only if $\Delta_i\in VT(s_j)$.
For a given query point $p$, $m''_p$, the number of triangles in $VT_S$ containing $p$, is between $m_p$ and $2m_p$. So, $m''_p$ gives a 2-approximation factor solution for VCP \cite{gud}. Since the visibility triangles of each segment may overlap, some of the segments are counted repeatedly.
In \cite{gud}, it is shown that each segment $s_i$ is counted $c_i$ times, where $c_i$ is the number of visible subsegments of $s_i$. In other words, there are $c_i$ triangles related to $s_i$ in $VT_S$ which contain $p$. 

A similar approach can be used to approximate $C(G(p))$.
A random subset $RVT_1\subset VT_S$ is chosen such that each member of $VT_S$ is chosen with probability $\frac{1}{m^{\beta}}$.
For a given query point $p$, let $c'_{i,1}\geq1$ be the number of triangles related to $s_i$ in $RVT_1$ containing $p$. We report $C_1=\sum^{n}_{i=1} (m^{\beta}c'_{i,1}-1)$ as the approximated value of $C(G(p))$ received by $RVT_1$. We choose $m^{\beta/2}$ random subsets $RVT_1,\dots,RVT_{m^{\beta/2}}$ of $VT_S$. Let $p$ be the given query point, for each $RVT_j$, $C_j=\sum^{n}_{i=1} (m^{\beta}c'_{i,j}-1)$ is calculated. At last, $C'_p=\frac{\sum^{m^{\beta/2}}_{j=1} C_j}{ m^{\beta/2}}$ is reported as the approximation value of $C(G(p))$. 

\subsection{Analysis of approximation factor}
We show that with the probability at least $\frac{1}{\log m}$, if $C(G(p))>\frac{1}{\delta^2}m^{\beta/2}\log m$, then $C'_p$ is a $(1+\delta)$-approximation of $C(G(p))$.
\begin{lemma}
$E(C_j)=C(G(p))$.
\end{lemma}
\begin{proof}
$E(C_j)=E(\sum^{n}_{i=1} m^{\beta}c'_{i,j}-1)=\sum^{n}_{i=1} E(m^{\beta}c'_{i,j}-1)=\sum^{n}_{i=1} c_{i}-1=C(G(p))$.
\end{proof}
Using Lemma~\ref{cheb}, we have,
\begin{center}
$\mathbb{P}=P(|\frac{C_1+\dots+C_{m^{\beta/2}}}{m^{\beta/2}}-C(G(p))|>\delta C(G(p)) )\leq \frac{Var(C_i)}{m^{\beta/2}\delta^2C(G(p))^2}$.
\end{center}
$Var(C_i)=m^{2\beta}C(G(p))(\frac{1}{m^{\beta}})(1-\frac{1}{m^{\beta}})$.
Since we have, $C(G(p))>\frac{1}{\delta^2}m^{\beta/2}\log m$,
\begin{center}
$\mathbb{P}=P(|\frac{C_1+\dots+C_{m^{\beta/2}}}{m^{\beta/2}}-C(G(p))|>\delta C(G(p)) )\leq \frac{1}{\log m}$.
\end{center}
So, with the probability at least $1-\mathbb{P}$, 
\begin{center}
$(1-\delta)C(G(p))\leq C'_p\leq (1+\delta)C(G(p))$.
\end{center}
and for a large $m$, we have, $\mathbb{P}\sim 0$.
\subsection{Analysis of time and space complexity}
By Lemma~\ref{arr}, for each $RVT_i$, a data structure of expected preprocessing time and size of $O(m^{2-2\beta}\log m)$ and $O(m^{2-2\beta})$ is needed. 
$RVT_i$ returns $C_i$ in $O(\log m)$ for each query point $p$. So, the expected space for all $m^{\beta/2}$ data structures is $O(m^{2-2\beta+\beta/2}\log m)$ and the query time for calculating $C'_p$ is $O(m^{\beta/2}\log m)$. So, we have the following lemma.
\begin{lemma}
\label{approx-components}
There exists an algorithm that approximates $C(G(p))$ in $O(\frac{1}{\delta^2}m^{\beta/2}\log m)$ query time by using $O(m^{2-3\beta/2})$ expected preprocessing time and $O(m^{2-3\beta/2})$ expected space $(0\leq \beta\leq \frac{2}{3})$. For each query $p$, this algorithm returns a value $C'_p$ such that with probability at least $1-\frac{1}{\log m}$, $(1-\delta)C(G(p))\leq C'_p\leq (1+\delta)C(G(p))$ when $C(G(p))> \frac{1}{\delta^2}m^{\beta/2}\log m$.
\end{lemma}

\section{Conclusion}
In this paper, a randomized algorithm is proposed to compute an approximation answer to VCP.
The main ideas of the algorithm that reduce the complexity of previous methods are random sampling and breaking the query into two steps. The time and space complexity of our algorithm depend on the size of $EVG(S)$.
A planar graph is associated to each query point $p$. It is proven that the answer is equal to $ve_p-C(G(p))$, where $ve_p$ is the number of visible end-points and $C(G(p))$ is the number of connected components in the planar graph.
To improve the running time of our algorithm instead of finding the exact values of $ve_p$ and $C(G(p))$, we approximate these values. Although an exact calculation of $ve_p$ using a tradeoff between the query time and the space is possible. 

\section*{Appendix}

\subsection*{Proof of Theorem~\ref{maint}}
Assume that we have the data structures of \cite{vet} and the algorithms of Lemma~\ref{approx-end-point} and Lemma~\ref{approx-components}.
For a given query point $p$,
we first run the algorithm of \cite{vet} for $\frac{2}{\delta^3}m^{\beta/2}\log m$ steps. If $VP_S(p)$ is calculated, then the exact value of 
$m_p$ is computed in $O(\frac{1}{\delta^3}m^{\beta/2}\log m)$. Otherwise, we calculate $ve'_p$ and $C'_p$ in $O(\frac{1}{\delta^2}m^{\beta/2}\log m)$ time by Lemma~\ref{approx-end-point} and Lemma~\ref{approx-components}. In the second case, we have $m_p>\frac{2}{\delta^3}m^{\beta/2}\log m$. As $m_p\leq 2ve_p$, so  $ve_p>\frac{1}{\delta^3}m^{\beta/2}$. Lemma~\ref{approx-end-point} also implies $|ve'_p-ve_p|\leq \delta ve_p$.

In the following it is shown that if $C'_p<(1+\delta)\frac{1}{\delta^2}m^{\beta/2}\log m$, then with probability at least $1-\frac{1}{\log m}$, we have $C(G(p))<\frac{1+\delta}{1-\delta}\frac{1}{\delta^2}m^{\beta/2}\log m$.

By Lemma~\ref{cheb} we have, 
\begin{center}
$p(|C(G(p)-C'_p|>\delta C(G(p))\leq \frac{m^{\beta/2}}{\delta^2C(G(p)}$

$p(C(G(p))>\frac{C'_p}{1-\delta})\leq \frac{m^{\beta/2}}{\delta^2 C(G(p))}$.
\end{center}
Which means, if $C(G(p))>\frac{1}{\delta^2}m^{\beta/2}\log m$, then with probability at most $\frac{1}{\log m}$ we have, $C(G(p))>\frac{C'_p}{1-\delta}$. So, with a probability at least $1-\frac{1}{\log m}$ we have 
\begin{center}
$C(G(p))\leq \frac{C'_p}{1-\delta}\leq \frac{1+\delta}{1-\delta}\frac{1}{\delta^2}m^{\beta/2}\log m$.
\end{center}
 Since, $C(G(p))\leq m_p$ and $m_p>\frac{1}{\delta^3}m^{\beta/2}\log m$, we have, $C(G(p))\leq \frac{1+\delta}{1-\delta}\delta m_p$. We know that $ve_p=m_p+C(G(p))$, and with probability at least $1-\frac{1}{\log m}$ we have, $(1-\delta)ve_p\leq ve'_p\leq (1+\delta)ve_p$, thus
\begin{center}
$\frac{ve'_p}{1+\delta}\leq m_p+C(G(p))\leq m_p+\frac{1+\delta}{1-\delta}\delta m_p$. 
\end{center}
Which implies
\begin{center}
$m_p\leq \frac{ve'_p}{1-\delta}\leq \frac{1+\delta^2(1+\delta)}{(1-\delta)^2}m_p$.
\end{center}
Let $1+\delta^*=\frac{1+\delta^2(1+\delta)}{(1-\delta)^2}$, then
\begin{center}
$m_p\leq \frac{ve'_p}{1-\delta}\leq (1+\delta^*)m_p$.
\end{center}
So, $\frac{ve'_p}{1-\delta}$ is reported as the approximated value of $m_p$.

If $C'_p>(1+\delta)\frac{1}{\delta^2}m^{\beta/2}\log m$, then according to Lemma~\ref{cheb}, we have,
\begin{center}
  $P(|C'_p-C(G(p)|\geq \frac{1}{\delta}m^{\beta/2}\log m )\leq \frac{m^{\beta/2} C(G(p))}{\frac{1}{\delta^2}m^{\beta}\log^2 m}$.
 \end{center} 
  So, $P(C(G(p))<C'_p-\frac{1}{\delta}m^{\beta/2}\log m)\leq \frac{m^{\beta/2} C(G(p))}{\frac{1}{\delta^2}m^{\beta}\log^2 m}$. We know that $C'_p>(1+\delta)\frac{1}{\delta^2}m^{\beta/2}\log m$. So, if $C(G(p))\leq \frac{1}{\delta^2}m^{\beta/2}\log m$, then $P(C(G(p))< \frac{1}{\delta^2}m^{\beta/2}\log m)\leq \frac{1}{\log m}$. We conclude that with probability at least $1-\frac{1}{\log m}$, $C(G(p))>\frac{1}{\delta^2}m^{\beta/2}\log m$, and so we can use Lemma~\ref{approx-components}.
\begin{center}
$m_p=ve_p-C(G(p))\leq \frac{ve'_p}{1-\delta}-\frac{C'_p}{1+\delta}$

$\leq \frac{(1+\delta)ve_p}{1-\delta}-\frac{(1-\delta)C(G(p))}{1+\delta}$

$\leq ve_p-C(G(p))+\frac{2(\delta) ve_p}{1-\delta}+\frac{2\delta C(G(p))}{1+\delta}$.
\end{center}
We know that $C(G(p))\leq m_p$. Moreover, we have, $m_p\leq 2ve_p$, so
\begin{center}
$\leq m_p+\frac{4\delta m_p}{1-\delta}+\frac{2\delta m_p}{1+\delta} $

$\leq (1+\frac{4\delta}{1-\delta}+\frac{2\delta}{1+\delta})m_p$.
\end{center}
Let $\delta^*=\frac{4\delta}{1-\delta}+\frac{2\delta}{1+\delta}$, then
\begin{center}
$m_p\leq \frac{ve'_p}{1-\delta}-\frac{C'_p}{1+\delta}\leq (1+\delta^*)m_p$.
\end{center}
Therefore, $\frac{ve'_p}{1-\delta}-\frac{C'_p}{1+\delta}$ is reported as the approximated value of $m_p$.
Note that $\delta^*<1$ and it can be arbitrary small by choosing $\delta$ small enough.

The query time of our algorithm is $O(\frac{1}{\delta^3}m^{\beta/2}\log m)$,  where the dependence of the variable $\delta^*$ to $\delta$ is as follows. When $\delta$ is less than a fixed constant $C$, $\delta^*$ is at most a linear fixed multiple of $\delta$ and hence, the query time of the algorithm can be expressed as $O(\frac{1}{\delta^3}m^{\beta/2}\log m)$. Note that for $\delta>C$ since $\delta^{-3}<C^{-3}$ it will be absorbed in the constant hidden in $O(m^{\beta/2}\log m)$.
\end{document}